%% file: arxiv.tex
\documentclass{article}

\usepackage[textwidth=12.75cm]{geometry}
\usepackage[dvipsnames]{xcolor}
\usepackage{verbatim}
\usepackage{paralist}

\usepackage{pgfplots}
\pgfplotsset{compat=1.10}

\usepackage[textsize=scriptsize]{todonotes}
\usepackage{colortbl}
\definecolor{Gray}{gray}{0.90}

\usepackage{amsmath}
\usepackage{amssymb}
\usepackage{amsthm}

\newtheorem{lemma}{Lemma}
\newtheorem{proposition}{Proposition}

\DeclareMathOperator*{\argmax}{argmax}

\DeclareMathOperator*{\E}{E}

\newcommand{\calU}[0]{\ensuremath{\mathcal{U}}}
\newcommand{\Real}[0]{\ensuremath{\mathbb{R}}}

\newcommand{\vb}[0]{\ensuremath{\boldsymbol{b}}}
\newcommand{\vp}[0]{\ensuremath{\boldsymbol{p}}}

\newif\ifExtendedVersion
\ExtendedVersiontrue

\usepackage{url}
\usepackage[breaklinks=true]{hyperref}
\usepackage[anythingbreaks]{breakurl}

\title{On the Economics of Ransomware}

\author{%
Aron Laszka\\
\texttt{aron.laszka@vanderbilt.edu}
\\ Vanderbilt University 
\and Sadegh Farhang\\
\texttt{farhang@ist.psu.edu}\\
Pennsylvania State University
\and Jens Grossklags\\
\texttt{jens.grossklags@in.tum.de}\\
Technical University of Munich}

\date{July, 2017\\Updated: August, 2017}

\begin{document}

\maketitle

\begin{abstract}
While recognized as a theoretical and practical concept for over 20 years, only now ransomware has taken centerstage as one of the most prevalent cybercrimes. Various reports demonstrate the enormous burden placed on companies, which have to grapple with the ongoing attack waves. At the same time, our strategic understanding of the threat and the adversarial interaction between organizations and cybercriminals perpetrating ransomware attacks is lacking.

In this paper, we develop, to the best of our knowledge, the first game-theoretic model of the ransomware ecosystem. Our model captures a multi-stage scenario involving organizations from different industry sectors facing a sophisticated ransomware attacker. We place particular emphasis on the decision of companies to invest in backup technologies as part of a contingency plan, and the economic incentives to pay a ransom if impacted by an attack. We further study to which degree comprehensive industry-wide backup investments can serve as a deterrent for ongoing attacks.

\end{abstract}


\input{intro.tex}

\input{model.tex}

\input{analysis.tex}

\input{numerical.tex}

\input{related.tex}

\input{concl.tex}

\textbf{Acknowledgments:} We thank the anonymous reviewers for their comments. The
research activities of Jens Grossklags are supported by the German
Institute for Trust and Safety on the Internet (DIVSI). Aron Laszka's work
was supported in part by the National Science Foundation (CNS-1238959) and
the Air Force Research Laboratory (FA 8750-14-2-0180).

\bibliographystyle{plain}
\bibliography{references}

\appendix
\input{proof.tex}

\end{document}

%% file: intro.tex
\section{Introduction}
\label{sec:intro}

Already in 1996, Young and Yung coined the term \textit{cryptovirological attacks} and provided a proof-of-concept implementation of what could now be considered a major building block of ransomware malware \cite{Young96}. Due to the perceived seriousness of this attack approach, they also suggested that ``access to cryptographic tools should be well controlled.'' 

Malware featuring ransomware behavior was at first deployed at modest scale (e.g., variants of PGPCoder/GPCode between approximately 2005-2010), and often suffered from technical weaknesses, which even led a researcher in the field to proclaim that ``ransomware as a mass extortion mean is certainly doomed to failure'' \cite{gazet2010comparative}. However, later versions of GPCode already used 1024-bit RSA key encryption; a serious threat even for well-funded organizations.

Ransomware came to widespread prominence with the CryptoLocker attack in 2013, which utilized Bitcoin as a payment vehicle \cite{Liao16}. Since then, the rise of ransomware has been dramatic, culminating (so far) with the 2017 attack waves of the many variants of the WannaCrypt/WannaCry and the Petya ransomwares. Targets include all economic sectors and devices ranging from desktop computers, entire business networks, industrial facilities, and also mobile devices. Security industry as well as law enforcement estimates for the amount of money successfully extorted and the (very likely much larger) overall damage caused by ransomware attacks differ widely. However, the figures are significant (see~Section~\ref{sec:related}).
Observing these developments, in a very recent retrospective article, Young and Yung bemoan the lack of adequate response focused on ransomware attacks by all stakeholders even though the threat was known for over 20 years~\cite{Young17}.

As with any security management decision, there is a choice between doing nothing to address a threat, or selecting an appropriate investment level. In the case of responding to a sophisticated ransomware attack, this primarily concerns decisions on how to invest in backup and recovery technologies, and whether to pay a ransom in  case of a successful attack. These decisions are interrelated.

The empirical evidence is mixed (and scarce). It is probably fair to say that backup technologies have always been somewhat of a stepchild in the overall portfolio of security technologies. In 2001, a survey showed that only 41\% of the respondents did data backups and 69\% had not recently facilitated a backup; at the same time, 25\% reported to have lost data \cite{Bruskin01}. In 2009, another survey found backup usage of less than 50\%; and 66\% reported to have lost files (42\% within the last 12 months) \cite{Kabooza09}. In a backup awareness survey that has been repeated annually since 2008, the figures for individuals who \textit{never} created backups have been slowly improving. Starting at 38\% in 2008, in the most recent survey in June 2017 only 21\% reported to have never made a backup. Still, only 37\% now report to create at least monthly backups \cite{Backblaze17}, despite the heightened media attention given to ransomware.

Regarding ransom payment behavior, IBM surveyed 600 business leaders in the U.S about ransomware, and their data management practices and perceptions. Within their sample, almost 50\% of the business representatives reported ransomware attacks in their organizations. Interestingly, 70\% of these executives reported that ransom payments were made in order to attempt a resolution of the incident. About 50\% paid over \$10,000 and around 20\% reported money transfers of over \$40,000 \cite{IBM16}. In contrast, a different survey of ransomware-response practices found that 96\% of those affected (over the last 12 months) did \textit{not} pay a ransom \cite{KnowBe4}. However, the characteristics of the latter sample are not described \cite{KnowBe4}. Finally, recent cybercrime measurement studies have tracked the approximate earnings for particular ransomware campaigns (for example, by tracking related Bitcoin wallets). These studies typically do not succeed in pinpointing the percentage of affected individuals or organizations paying the ransom (e.g., \cite{kharraz2015cutting,Liao16}).

Our work targets two key aspects of a principled response to sophisticated ransomware attacks. First, we develop an economic model to further our strategic understanding of the adversarial interaction between organizations attacked by ransomware and ransomware attackers. As far as we know, our work is the first such game-theoretic model. 

Second, we study the aforementioned response approaches to diminish the economic impact of the ransomware threat on organizations. As such our model focuses on organizations' decision-making regarding backup investments (as part of an overall contingency plan), which is an understudied subject area. We further determine how backup security investments interact with an organization's willingness to pay a ransom in case of a ransomware attack. 

Further, we numerically show how (coordinated) backup investments by organizations can have a deterrent effect on ransomware attackers. Since backup investments are a private good and are not subject to technical interdependencies, this observation is novel to the security economics literature and relatively specific to ransomware. Note, for example, that in the context of cyberespionage and data breaches to exfiltrate data, such a deterrence effect of backup investments is unobservable.

We proceed as follows. In Section~\ref{sec:model}, we develop our game-theoretic model. We conduct a thorough analysis of the model in Section~\ref{sec:analysis}. In Section~\ref{sec:numerical}, we complement our analytic results with a numerical analysis. We discuss additional related work on ransomware as well as security economics in Section~\ref{sec:related}, and offer concluding remarks in Section~\ref{sec:concl}.

%
%
%
%
%

%% file: model.tex
\section{Model}
\label{sec:model}

We model ransomware attacks as a multi-stage, multi-defender security game.
Table~\ref{tab:symbols} shows a list of the symbols used in our model.

\begin{table}[h!]
\centering
\caption{List of Symbols}
\label{tab:symbols}
\begin{tabular}{|c|l|}
\hline
Symbol & Description \\ 
\hline
$G_j$ & set of organizations belonging to group $j$ \\
\rowcolor{Gray} $W_j$ & initial wealth of organizations in group $j$ \\
$\beta$ & discounting factor for uncertain future losses \\
\rowcolor{Gray} $F_j$ & cost of data loss due to random failures in group $j$ \\
$L_j$ & cost of permanent data loss due to ransomware attacks in group $j$ \\
\rowcolor{Gray} $T_j$ & loss from business interruptions due to ransomware in group $j$ \\
$D$ & base difficulty of perpetrating ransomware attacks \\
\rowcolor{Gray} $C_B$ & unit cost of backup effort \\
$C_A$ & unit cost of attack effort \\
\rowcolor{Gray} $C_D$ & fixed cost of developing ransomware \\
$b_i$ & backup effort of organization $i$ \\
\rowcolor{Gray} $p_i$ & decision of organization $i$ about ransom payment \\
$a_j$ & attacker's effort against group $j$ \\
\rowcolor{Gray} $r$ & ransom demanded by the attacker \\
$V_j(a_1, a_2)$ & probability of an organization $i \in G_j$ becoming compromised \\
\hline
\end{tabular}
\end{table}

\subsection{Players}

On the defenders' side, players model organizations that are susceptible to ransomware attacks.
Based on their characteristics, we divide these organizations into two groups (e.g., hospitals and universities).
We will refer to these two groups as group $1$ and group $2$, and we let set $G_1$ and set $G_2$ denote their members, respectively.
On the attacker's side, there is a single player, who models cybercriminals that may develop and deploy ransomware.
Note that we model attackers as a single entity since our goal is to understand and improve the behavior of defenders; hence, competition between attackers is not our current focus. 
Our model---and many of our results---could be extended to multiple attackers in a straightforward manner. 

\subsection{Strategy Spaces}

With our work, we focus on the mitigation of ransomware attacks through backups (as a part of contingency plans), and we will not consider the organizations' decisions on preventative effort (e.g., firewall security policies). The tradeoff between mitigation and preventative efforts has been subject of related work \cite{grossklags2008secure}. We let~$b_i \in \Real_+$ denote the backup effort of organization $i$, which captures the frequency and coverage of backups as well as contingency plans and preparations. Compromised organizations also have to decide whether they pay the ransom or sustain permanent data loss.  We let $p_i = 1$ if organization $i$ pays, and $p_i = 0$ if it does not pay.

The attacker first decides whether it wishes to engage in cybercrime using ransomware.
If the attacker chooses to engage, then it has to select the amount of effort spent on perpetrating the attacks.
We let $a_1 \in \Real_{\geq 0}$ and $a_2 \in \Real_{\geq 0}$ denote the attacker's effort spent on attacking group $1$ and group $2$, respectively.
If the attacker chooses not to attack group $j$ (or not to engage in cybercrime at all), then $a_j = 0$.
We assume that each organization within a group falls victim to the attack with the same probability $V_j(a_1, a_2)$, which depends on the attacker's effort, independently of the other organizations.
Since the marginal utility of attack effort is typically decreasing, we assume that the infection probability $V_j(a_1, a_2)$ is
\begin{equation}
V_j(a_1, a_2) = \frac{a_j}{D + (a_1 + a_2)} ,
\end{equation}
where $D$ is the base difficulty of attacks.
In the formula above, the numerator expresses that as the attacker increases its effort on group $j$, more and more organizations fall victim.
Meanwhile, the denominator captures the decreasing marginal utility: as the attacker increases its attack effort, compromising additional targets becomes more and more difficult.
In practice, this corresponds to the increasing difficulty of finding new targets as organizations are becoming aware of a widespread ransomware attack and are taking precautions, etc.

The attacker also has to choose the amount of ransom $r$ to demand from compromised organizations in exchange for restoring their data and systems.

\subsubsection{Stages}
The game consists of two stages:
\begin{itemize}
\item Stage I: Organizations choose their backup efforts~$\vb$, while the attacker chooses its attack effort $a_1$ and $a_2$, as well as its ransom demand $r$.
\item Stage II: Each organization $i \in G_j$ becomes compromised with probability $V_j(a_1, a_2)$. Then, organizations that have fallen victim to the attack choose whether to pay the ransom or not, which is represented by $\vp$.
\end{itemize}

\subsection{Payoffs}

\subsubsection{Defender's Payoff}
If an organization $i$, which belongs to group $j \in \{1, 2\}$, has not fallen victim to a ransomware attack, then its payoff is
\begin{equation}
\calU_{O_i}\big|_{\text{not compromised}} = W_j - C_B \cdot b_i - \beta \frac{F_j}{b_i} ,
\label{eq:payoffSecureOrg}
\end{equation}
where $W_j$ is the initial wealth of organizations in group $j$, $F_j$ is their loss resulting from corrupted data due to random failures\footnote{Since we interpret effort $b_i$ primarily as the frequency of backups, the fraction $\frac{1}{b_i}$ is proportional to the expected time since the last backup. Consequently, we assume that data losses are inversely proportional to $b_i$.
Note that alternative interpretations, such as assuming $b_i$ to be the level of sophistication of backups (e.g., air-gapping), which determines the probability that the backups remain uncompromised, also imply a similar relationship.}, 
and $C_B$ is the unit cost of backup effort. 
The parameter $\beta$ is a behavioral discount factor, which captures the robust empirical observation that individuals underappreciate the future consequences of their current actions \cite{Rabin99}. The magnitude of $\beta$ is assumed to be related to underinvestment in security and privacy technologies \cite{Acquisti07,Grossklags14}; in our case, procrastination of backup investments \cite{Baddeley11}.\footnote{We are unaware of any behavioral study that specifically investigates the impact of the \textit{present bias} behavioral discount factor on backup decisions, but industry experts argue strongly for its relevance. For example, in the context of the 2017 WannaCry ransomware attacks a commentary about backups stated: ``This may be stating the obvious, but it's still amazing to know the sheer number of companies that keep procrastinating over this important task \cite{Venkat17}.''}

Otherwise, we have two cases.
If organization $i$ decides to pay the ransom $r$, then its payoff is
\begin{equation}
W_j - C_B \cdot b_i - \beta \left( \frac{F_j}{b_i} + T_j + r \right), \nonumber
\end{equation}
where $T_j$ is the loss resulting from temporary business interruption due to the attack.
On the other hand, if organization $i$ does not pay the ransom, then its payoff is
\begin{equation}
W_j - C_B \cdot b_i - \beta \left( \frac{F_j + L_j}{b_i} + T_j \right) , \nonumber
\end{equation}
where $L_j$ is the loss resulting from permanent data loss due to the ransomware attack. 
Using $p_i$, we can express a compromised organization's payoff as
\begin{equation}
\calU_{O_i}\big|_{\text{compromised}} = W_j - C_B \cdot b_i - \beta \left( \frac{F_j + (1 - p_i) \cdot L_j}{b_i} + T_j + p_i \cdot r \right).
\label{eq:payoffCompromisedOrg}
\end{equation}

By combining Equations~\eqref{eq:payoffSecureOrg} and~\eqref{eq:payoffCompromisedOrg} with $V_j$, we can express the expected utility of an organization $i \in G_j$ as
\begin{align}
\E\left[\calU_{O_i}\right] = &\left(1 - V_j(a_1, a_2)\right) \left[ W_j - C_B \cdot b_i - \beta \frac{F_j}{b_i} \right] \nonumber\\
&+ V_j(a_1, a_2) \left[ W_j - C_B \cdot b_i - \beta \left( \frac{F_j + (1 - p_i) \cdot L_j}{b_i} + T_j + p_i \cdot r \right) \right] .
\end{align}

\subsubsection{Attacker's Payoff}
For the attacker's payoff, we also have two cases.
If the attacker decides not to participate (i.e., if $a_1 = 0$ and $a_2 = 0$), then its payoff is simply zero.
Otherwise, its payoff depends on the number of organizations that have fallen victim and decided to pay.
We can calculate the expected number of victims who pay the ransom as 
\begin{equation}
\E[\text{number of victims who pay the ransom}] = \sum_j \sum_{i \in G_j} V_j(a_1, a_2)  \cdot p_i 
\end{equation}
since each organization $i \in G_j$ is compromised with probability $V_j$, and $p_i = 1$ if organization $i$ chooses to pay (and $p_i = 0$ if it does not pay).

Then, we can express the attacker's expected payoff simply as
\begin{equation}
\E\left[\calU_A\right] = \left[\sum_j \sum_{i \in G_j} V_j(a_1, a_2)\cdot p_i \right] \cdot r - C_A \cdot (a_1 + a_2) - C_D , 
\end{equation}
where $C_A$ is the unit cost of attack effort, and $C_D$ is the fixed cost of developing a ransomware, which the attacker must pay if it decides to engage (i.e., if $a_1 > 0$ or $a_2 > 0$).

\subsection{Solution Concepts}
\label{sub:solCon}

We assume that every player is interested in maximizing its expected payoff, and we use subgame perfect Nash equilibrium as our solution concept.
We also assume that organizations always break ties (i.e., when both paying and not paying are best responses) by choosing to pay.
Note that the latter assumption has no practical implications, it only serves to avoid pathological mathematical~cases.

Further, in our numerical analysis in Section~\ref{sec:numerical}, we will use the \textit{social optimum} concept for comparison with the Nash equilibrium results. 
In the social optimum, a social planner can coordinate the decisions of organizations, such that it yields the maximum aggregate outcome for the organizations, subject to an optimal response by the attacker (who is not guided by the social planner).


%% file: analysis.tex
\section{Analysis}
\label{sec:analysis}

In this section, we analyze our proposed game-theoretic model of the ransomware ecosystem. Our solution concept, as mentioned in Section~\ref{sub:solCon}, is the subgame perfect Nash equilibrium. Hence, in our analysis, we first calculate each organization's decision in Stage II in Section~\ref{sub:stage2}. In other words, we derive under what conditions a victim organization will pay the requested ransom from the attacker.  Then, we calculate the best-response backup strategy for each organization in Stage I of the game in Section~\ref{sub:Backup}. Third, we calculate the attacker's best-response, i.e., demanded ransom and the attacker's effort, in Section~\ref{sub:AttBR}. By calculating the attacker's and the organizations' best-responses, we can then derive the Nash equilibrium in Section~\ref{sub:equilibrium}. 

\subsection{Compromised Organizations' Ransom Payment Decisions}
\label{sub:stage2}

We begin our analysis by studying the compromised organizations' best-response payment strategies in the second stage of the game.

\begin{lemma}
\label{lem:paymentBestResponse}
For organization $i \in G_j$, paying the ransom (i.e., $p_i = 1$) is a best response if and only if
\begin{equation}
r \leq \frac{L_j}{b_i} .
\end{equation}
\end{lemma}

Proof of Lemma~\ref{lem:paymentBestResponse} is provided 
in Appendix~\ref{app:proof1}. 

Lemma~\ref{lem:paymentBestResponse} means that an organization will pay the demanded ransom if the demanded value is not higher than the average permanent data loss due to ransomware attack.

\subsection{Organizations' Backup Decisions}
\label{sub:Backup}

We next study the organizations' best-response backup strategies in the first stage.
We assume that compromised organizations will play their best responses in the second stage (see Lemma~\ref{lem:paymentBestResponse}), but we do not make any assumptions about the attacker's effort or ransom strategies.
We first characterize the organizations' best-response backup strategies when they do not face any attacks (Lemma~\ref{lem:backupBestResponseNoAttack}) and then in the case when they are threatened by ransomware (Lemma~\ref{lem:backupBestResponseAttack}).

\begin{lemma}
\label{lem:backupBestResponseNoAttack}
If the attacker chooses not to attack group $j$ (i.e., $a_j = 0$), then the unique best-response backup strategy for organization $i \in G_j$ is
\begin{equation}
b_i^* = \sqrt{\beta\frac{F_j}{C_B}} .
\end{equation}
\label{lem:BackupNoAtt}
\end{lemma}

Proof of Lemma~\ref{lem:backupBestResponseNoAttack} is provided 
in Appendix~\ref{app:proof2}. 

Note that in Lemma~\ref{lem:backupBestResponseNoAttack}, an organization chooses its backup strategy by considering data loss due to random failures rather than data loss due to ransomware attack since that organization is not chosen to be attacked by the attacker. 

Lemma~\ref{lem:backupBestResponseAttack} calculates an organization's best-response backup strategy. Note that an organization chooses its backup strategy at Stage I. In this stage, an organization does not know whether it is the target of a ransomware attack and if that organization is the target of a ransomware attack, whether the attack is successful. 

\begin{lemma}
\label{lem:backupBestResponseAttack}
If the attacker chooses to attack group $j$ (i.e., $a_j > 0$), then the best-response backup strategy $b_i^*$ for organization $i \in G_j$ is
\begin{itemize}
\item if $b_j^{\text{low}} > \frac{L_j}{r}$, then $b_i^* = b_j^{\text{high}}$;
\item if $b_j^{\text{high}} < \frac{L_j}{r}$, then $b_i^* = b_j^{\text{low}}$;
\item otherwise, $b_i^* \in \left\{b_j^{\text{low}}, b_j^{\text{high}}\right\}$ (the one that gives the higher payoff or both if the resulting payoffs are equal), 
\end{itemize}
where $b_j^{\text{low}} = \sqrt{\beta\frac{F_j}{C_B}}$ and $b_j^{\text{high}} = \sqrt{\beta\frac{\left( F_j + V_j(a_1, a_2)  L_j \right)}{C_B}}$.
\end{lemma}

Proof of Lemma~\ref{lem:backupBestResponseAttack} is provided \ifExtendedVersion
in Appendix~\ref{app:proof3}. 
\else
online in the extended version of the paper~\cite{extended_version}.
\fi

Lemma~\ref{lem:backupBestResponseAttack} shows the best-response backup strategy when an organization is under attack. If the demanded ransom value is high, i.e., $r > \frac{L_j}{b_j^{\text{low}}}$, an organization takes into account the data loss due to ransomware attack as well as the data loss due to random failure when choosing the backup strategy level. On the other hand, if the demanded ransom is low, i.e., $r < \frac{L_j}{b_j^{\text{high}}}$, an organization does not care about the data loss due to ransomware attack even when that organization is under attack. In other words, that organization behaves like an organization that is not under ransomware attack, i.e., similar to Lemma~\ref{lem:backupBestResponseNoAttack}. 

\subsection{Attacker's Best Response}
\label{sub:AttBR}

Building on the characterization of the organizations' best responses, we now characterize the attacker's best-response strategies.
Notice that the lemmas presented in the previous section show that an organization's best response does not depend on the identity of the organization, only on its group.
Since we are primarily interested in studying equilibria, in which everyone plays a best response, we can make the following assumptions:
\begin{itemize}
\item All organizations within a group $j$ play the same backup strategy, which is denoted by $\hat{b}_j$.
\item $\frac{L_1}{\hat{b}_1} \leq \frac{L_2}{\hat{b}_2}$.
\end{itemize}
The second assumption is without loss of generality since we could easily re-number the groups.

In Lemma~\ref{lem:ransomBestResponse}, we calculate the attacker's best-response demanded ransom given the attacker's effort and the organizations' backup strategies. 

\begin{lemma}
\label{lem:ransomBestResponse}
If the attacker's effort $(a_1, a_2)$ is fixed, then its best-response ransom demand $r^*$ is 
\begin{itemize}
\item $\frac{L_1}{\hat{b}_1}$ if $\left|G_1\right|  V_1(a_1, a_2) \frac{L_1}{\hat{b}_1} > \left|G_2\right|  V_2(a_1, a_2) \left( \frac{L_2}{\hat{b}_2} -  \frac{L_1}{\hat{b}_1} \right)$
\item $\frac{L_2}{\hat{b}_2}$ if $\left|G_1\right|  V_1(a_1, a_2) \frac{L_1}{\hat{b}_1} < \left|G_2\right|  V_2(a_1, a_2) \left( \frac{L_2}{\hat{b}_2} -  \frac{L_1}{\hat{b}_1} \right)$
\item both $\frac{L_1}{\hat{b}_1}$ and $\frac{L_2}{\hat{b}_2}$ otherwise.
\end{itemize}
\label{Lem:AttBR-Ransom}
\end{lemma}

Proof of Lemma~\ref{lem:ransomBestResponse} is provided 
in Appendix~\ref{app:proof4}. 



Lemma~\ref{lem:AttBREffort} shows how the attacker divides its best-response attack effort between the two groups of organizations. Here, we assume that $a_1+a_2 = a_{sum}$, where $a_{sum}$ is a constant. Note that it is possible that the attacker decides not to attack either of the groups of organizations. The reason is that the benefit for the attacker from a ransomware attack may be lower than the cost of the attack. Hence, a rational attacker will abstain from attacking either of the groups. 

\begin{lemma}
	The attacker's best-response attack effort $(a_1^*, a_2^*)$ is as follows:
	
	\begin{itemize}
		\item $a_1^* = 0$ and $a_2^* = a_{sum}$ if $|G_1| \cdot 1_{\left\{r \leq \frac{L_1}{\hat{b}_1}\right\}} < |G_2| \cdot 1_{\left\{r \leq \frac{L_2}{\hat{b}_2}\right\}}$ and $\frac{a_{\text{sum}} }{D + a_{\text{sum}}} |G_2| \cdot r \cdot 1_{\left\{r \leq \frac{L_2}{\hat{b}_2}\right\}} > C_A \cdot a_{sum} + C_D$,
		
		\item $a_1^* = a_{\text{sum}}$ and $a_2^*=0$ if $|G_1| \cdot 1_{\left\{r \leq \frac{L_1}{\hat{b}_1}\right\}} > |G_2| \cdot 1_{\left\{r \leq \frac{L_2}{\hat{b}_2}\right\}}$ and $\frac{a_{\text{sum}} }{D + a_{\text{sum}}} |G_1| \cdot r \cdot 1_{\left\{r \leq \frac{L_2}{\hat{b}_2}\right\}} > C_A \cdot a_{sum} + C_D$,
		
		\item any $a_1^*$ between $0$ and $a_{\text{sum}}$ and $a_2^*= a_{sum}-a_1^*$ if $|G_1| \cdot 1_{\left\{r \leq \frac{L_1}{\hat{b}_1}\right\}} = |G_2| \cdot 1_{\left\{r \leq \frac{L_2}{\hat{b}_2}\right\}}$ and $\frac{a_{\text{sum}} }{D + a_{\text{sum}}} |G_2| \cdot r \cdot 1_{\left\{r \leq \frac{L_2}{\hat{b}_2}\right\}} > C_A \cdot a_{sum} + C_D$.
		
		\item $a_1^* = a_2^*=0$ otherwise. 
	\end{itemize}
	\label{lem:AttBREffort}
\end{lemma}

Proof of Lemma~\ref{lem:AttBREffort} is provided 
in Appendix~\ref{app:proof5}. 

\subsection{Equilibria}
\label{sub:equilibrium}

Proposition~\ref{pro:NE} provides the necessary and sufficient conditions for the attacker's strategy to abstain from attack, i.e., $a_1^*=a_2^*=0$, and $\hat{b}_1^*=\sqrt{\beta\frac{F_1}{C_B}}$ and $\hat{b}_2^*=\sqrt{\beta\frac{F_2}{C_B}}$ is Nash equilibrium. 

\begin{proposition}
The attacker choosing \emph{not to attack} and the organizations choosing backup efforts $\sqrt{\beta\frac{F_1}{C_B}}$ and $\sqrt{\beta\frac{F_2}{C_B}}$ is an equilibrium if and only if each of the following conditions are satisfied:
\begin{itemize}
\item $ \frac{L_2 \cdot a_{sum} \cdot |G_2| \cdot 1_{\left\{r \leq \frac{L_2}{\hat{b}_2}\right\}}}{L_2\left(D + a_{sum}\right) \left(C_A \cdot a_{sum} + C_D\right)} < \sqrt{\beta\frac{F_2}{C_B}}$ and $|G_1| \cdot 1_{\left\{r \leq \frac{L_1}{\hat{b}_1}\right\}} \leq |G_2| \cdot 1_{\left\{r \leq \frac{L_2}{\hat{b}_2}\right\}}$

\item $ \frac{L_2 \cdot a_{sum} \cdot |G_1| \cdot 1_{\left\{r \leq \frac{L_1}{\hat{b}_1}\right\}}}{L_2\left(D + a_{sum}\right) \left(C_A \cdot a_{sum} + C_D\right)} < \sqrt{\beta\frac{F_2}{C_B}}$ and $|G_1| \cdot 1_{\left\{r \leq \frac{L_1}{\hat{b}_1}\right\}} > |G_2| \cdot 1_{\left\{r \leq \frac{L_2}{\hat{b}_2}\right\}}$
\end{itemize}
\label{pro:NE}
\end{proposition}

Proof of Proposition~\ref{pro:NE} is provided 
in Appendix~\ref{app:proofPro1}.

%% file: numerical.tex
\section{Numerical Illustrations}
\label{sec:numerical}

In this section, we present numerical results on our model.
We first compare equilibria to social optima, and we study the effect of changing the values of key parameters (Section~\ref{sec:numEqSocOpt}).
We then investigate interdependence between multiple groups of organizations, which is caused by the strategic nature of attacks, and we again study the effect of changing key parameters (Section~\ref{sec:numInter}).

For any combination of parameter values, our game has at most one equilibrium, which we will plot in the figures below.
However, for some combinations, the game does not have an equilibrium.
In these cases, we used iterative best responses:
\begin{compactenum}
\item starting from an initial strategy profile, 
\item we changed the attacker's strategy to a best response, 
\item we changed the organization's strategy to a best response, 
\item and then we repeated from Step 2.
\end{compactenum}
We found that regardless of the initial strategy profile, the iterative best-response dynamics end up oscillating between two strategy profiles.
Since these strategy profiles were very close, we plotted their averages in place of the equilibria in the figures below.

\subsection{Equilibria and Social Optima}
\label{sec:numEqSocOpt}

For clarity of presentation, we consider a single organization type in this subsection.
%
The parameter values used in this study are as follows: $|G|=100$, $W=100$, $\beta=0.9$, $F=5$, $L=5$, $T=10$, $C_B=1$, $D=10$, $C_A=10$, and $C_D=10$ (unless stated otherwise).

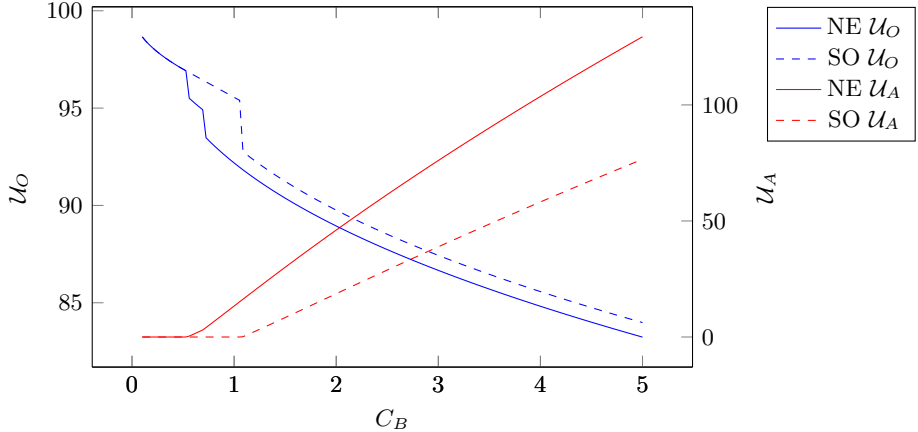
\begin{figure}[h!]
\centering
\begin{tikzpicture}[font=\small]
\begin{axis}[
  width=0.75\textwidth,
  height=0.5\textwidth,
  axis y line*=right,
  ylabel=$\calU_A$,
]
\addplot[no markers, solid, red] table[x=C_B, y=NE_att_payoff, col sep=comma, comment chars={\%}] {results/C_B.csv};
\addplot[no markers, dashed, red] table[x=C_B, y=SO_att_payoff, col sep=comma, comment chars={\%}] {results/C_B.csv};
\end{axis}
\begin{axis}[
  legend pos=outer north east,
  legend style={shift={(0.75cm,0cm)}},
  axis y line*=left,
  width=0.75\textwidth,
  height=0.5\textwidth,
  xlabel=$C_B$,
  ylabel=$\calU_O$,
]
\addplot[no markers, solid, blue] table[x=C_B, y=NE_org_payoff, col sep=comma, comment chars={\%}] {results/C_B.csv};
\addlegendentry{NE $\calU_O$}; 
\addplot[no markers, dashed, blue] table[x=C_B, y=SO_org_payoff, col sep=comma, comment chars={\%}] {results/C_B.csv};
\addlegendentry{SO $\calU_O$}; 
\addlegendimage{solid, red}
\addlegendentry{NE $\calU_A$}
\addlegendimage{dashed, red}
\addlegendentry{SO $\calU_A$}
\end{axis}
\end{tikzpicture}
\caption{The expected payoff of the attacker (\textcolor{red}{red}) and an individual organization (\textcolor{blue}{blue}) in equilibrium (solid line ---) and in social optimum (dashed line - - -) for various backup cost values $C_B$.}
\label{fig:C_B}
\end{figure}

Figure~\ref{fig:C_B} shows the expected payoffs of an individual organization and the attacker for various values of the unit cost $C_B$ of backup effort.
In practice, the unit cost of backup effort may change, for example, due to technological improvements (decreasing the cost) or growth in the amount of data to be backed up (increasing the cost).
When this cost is very low ($C_B < 0.5$), organizations can perform frequent and sophisticated backups, which means that the amount of data that may be compromised---and hence, the ransom that they are willing to pay---is very low.
As a result, the attacker is deterred from deploying ransomware ($\calU_A = 0$) since its income from ransoms would not cover its expenses.
For higher costs ($0.5 \leq C_B < 1$), the organizations' equilibrium payoff is much lower since they choose to save on backups, which incentivizes the attacker to deploy ransomware and extort payments from them.
In this case, the social optimum for the organizations is to maintain backup efforts and, hence, deter the attacker.
For even higher costs ($C_B \geq 1$), deterrence is not socially optimal.
However, the equilibrium payoffs are still lower since organizations shirk on backup efforts, which leads to more intense attacks and higher ransom demands.

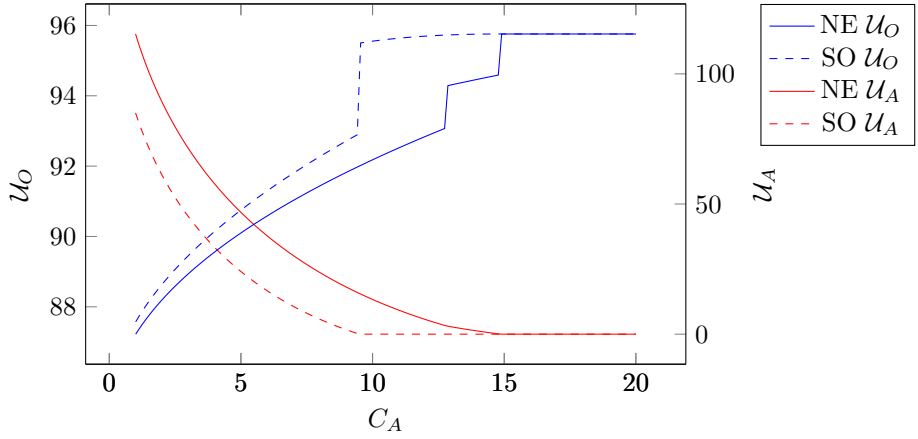
\begin{figure}[h!]
\centering
\begin{tikzpicture}
\begin{axis}[
  height=0.5\textwidth,
  width=0.75\textwidth,
  axis y line*=right,
  ylabel=$\calU_A$,
]
\addplot[no markers, solid, red] table[x=C_A, y=NE_att_payoff, col sep=comma, comment chars={\%}] {results/C_A.csv};
\addplot[no markers, dashed, red] table[x=C_A, y=SO_att_payoff, col sep=comma, comment chars={\%}] {results/C_A.csv};
\end{axis}
\begin{axis}[
  legend pos=outer north east,
  legend style={shift={(0.75cm,0cm)}},
  axis y line*=left,
  height=0.5\textwidth,
  width=0.75\textwidth,
  xlabel=$C_A$,
  ylabel=$\calU_O$,
]
\addplot[no markers, solid, blue] table[x=C_A, y=NE_org_payoff, col sep=comma, comment chars={\%}] {results/C_A.csv};
\addlegendentry{NE $\calU_O$}; 
\addplot[no markers, dashed, blue] table[x=C_A, y=SO_org_payoff, col sep=comma, comment chars={\%}] {results/C_A.csv};
\addlegendentry{SO $\calU_O$}; 
\addlegendimage{solid, red}
\addlegendentry{NE $\calU_A$}
\addlegendimage{dashed, red}
\addlegendentry{SO $\calU_A$}
\end{axis}
\end{tikzpicture}
\caption{The expected payoff of the attacker (\textcolor{red}{red}) and an individual organization (\textcolor{blue}{blue}) in equilibrium (solid line ---) and in social optimum (dashed line - - -) for various attack cost values $C_A$.}
\label{fig:C_A}
\end{figure}

Figure~\ref{fig:C_A} shows the expected payoff of an individual organization and the attacker for various values of the unit cost $C_A$ of attack effort.
In practice, the unit cost of attack effort can change, e.g., due to the development of novel attacks and exploits (lowering the cost) or the deployment of more effective defenses (increasing the cost).
Figure~\ref{fig:C_A} shows phenomena that are similar to the ones exhibited in Figure~\ref{fig:C_B}.
When the attacker is at a technological advantage (i.e., when $C_A$ is low), deterrence is not a realistic option for organizations.
However, they can improve their payoffs---compared to the equilibrium---by coordinating and investing more in backups, thereby achieving social optimum.
For higher attack costs ($10 < C_A \leq 15$), this coordination can result in significantly higher payoffs since deterrence becomes a viable option.
For very high attack costs ($C_A > 15$), compromising an organization costs more than what the attacker could hope to collect with ransoms; hence, coordination is no longer necessary to deter the attacker.  

\begin{figure}[h!]
\centering
\begin{tikzpicture}
\begin{axis}[
  width=0.75\textwidth,
  height=0.5\textwidth,
  axis y line*=right,
  ylabel=$\calU_A$,
]
\addplot[no markers, solid, red] table[x=beta, y=NE_att_payoff, col sep=comma, comment chars={\%}] {results/beta.csv};
\addplot[no markers, dashed, red] table[x=beta, y=SO_att_payoff, col sep=comma, comment chars={\%}] {results/beta.csv};
\end{axis}
\begin{axis}[
  legend pos=outer north east,
  legend style={shift={(0.75cm,0cm)}},
  axis y line*=left,
  width=0.75\textwidth,
  height=0.5\textwidth,
  xlabel=$\beta$,
  ylabel=$b$,
]
\addplot[no markers, solid, OliveGreen] table[x=beta, y=NE_b, col sep=comma, comment chars={\%}] {results/beta.csv};
\addlegendentry{NE $b$}; 
\addplot[no markers, dashed, OliveGreen] table[x=beta, y=SO_b, col sep=comma, comment chars={\%}] {results/beta.csv};
\addlegendentry{SO $b$}; 
\addlegendimage{solid, red}
\addlegendentry{NE $\calU_A$}
\addlegendimage{dashed, red}
\addlegendentry{SO $\calU_A$}
\end{axis}
\end{tikzpicture}
\caption{The attacker's expected payoff (\textcolor{red}{red}) and the organizations' backup strategy (\textcolor{OliveGreen}{green}) in equilibrium (solid line ---) and in social optimum (dashed line - - -) for various discounting factor values $\beta$.}
\label{fig:beta}
\end{figure}
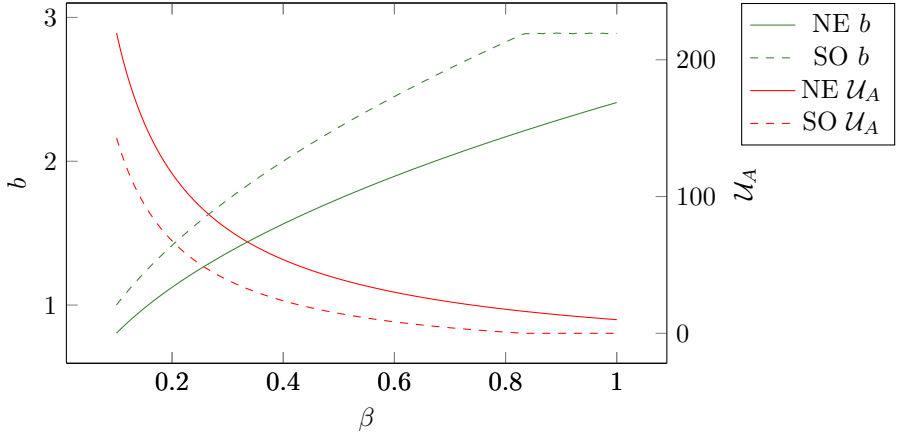

Figure~\ref{fig:beta} shows how the organizations' backup efforts $b$ and the attacker's payoff are effected by the behavioral discount factor $\beta$.
With low values of $\beta$, organizations underappreciate future consequences; hence, they shirk on backup efforts (as evidenced by low values of $b$).
With high values of $\beta$, organizations care more about future losses, so they invest more in backup efforts (resulting in high values of $b$).
We see that in all cases, there is a significant difference between the equilibrium and the social optimum.
This implies that regardless of the organizations' appreciation of future consequences, coordination is necessary.
In other words, low backup efforts cannot be attributed only to behavioral factors.

\subsection{Interdependence}
\label{sec:numInter}

Now, we study the interdependence between two groups of organizations.
We instantiate the parameters of both organizations (and the attacker) with the same values as in the previous subsection. Note that more numerical illustrations are available \ifExtendedVersion
in Appendix~\ref{app:Num}. 
\else
online in the extended version of the paper~\cite{extended_version}.
\fi

\pgfplotsset{figL/.style={%
  width=0.55\linewidth,
  xlabel={$L_1$},
  ylabel={$L_2$},
  thick,
  xtick={1, 5, 10},
  ytick={1, 5, 10},
  xlabel shift=-0.5em,
  ylabel shift=-0.5em,
  zlabel shift=-0.5em,
}}

\begin{figure}[h!]
\centering
\renewcommand*{\arraystretch}{0}
\setlength{\tabcolsep}{1.5pt}
\begin{tabular}{cc}
\begin{tikzpicture}[font=\small]
\begin{axis}[
  figL,
  view={100}{55},
]
\addplot3 [surf, mesh/ordering=x varies, mesh/rows=30, mesh/cols=30] table[x=L1, y=L2, z=NE_org1_payoff, col sep=comma, comment chars={\%}] {results/L.csv}; 
\end{axis}
\end{tikzpicture}
&
\begin{tikzpicture}
\begin{axis}[
  figL,
  view={170}{55},
]
\addplot3 [surf, mesh/ordering=x varies, mesh/rows=30, mesh/cols=30] table[x=L1, y=L2, z=NE_org2_payoff, col sep=comma, comment chars={\%}] {results/L.csv}; 
\end{axis}
\end{tikzpicture}
\\
(a) payoff for an organization of type 1 & (b) payoff for an organization of type 2 \\[1.2em]
\multicolumn{2}{c}{
\begin{tikzpicture}
\begin{axis}[
  figL,
  view={300}{55},
]
\addplot3 [surf, mesh/ordering=x varies, mesh/rows=30, mesh/cols=30] table[x=L1, y=L2, z=NE_att_payoff, col sep=comma, comment chars={\%}] {results/L.csv}; 
\end{axis}
\end{tikzpicture}
}
\\
\multicolumn{2}{c}{(c) payoff for the attacker}
\end{tabular}
\caption{The expected payoff of individual organizations of (a) type 1 and (b) type 2 as well as (c) the attacker in equilibrium for various data loss costs $L_1$ and $L_2$.}
\label{fig:L}
\end{figure}
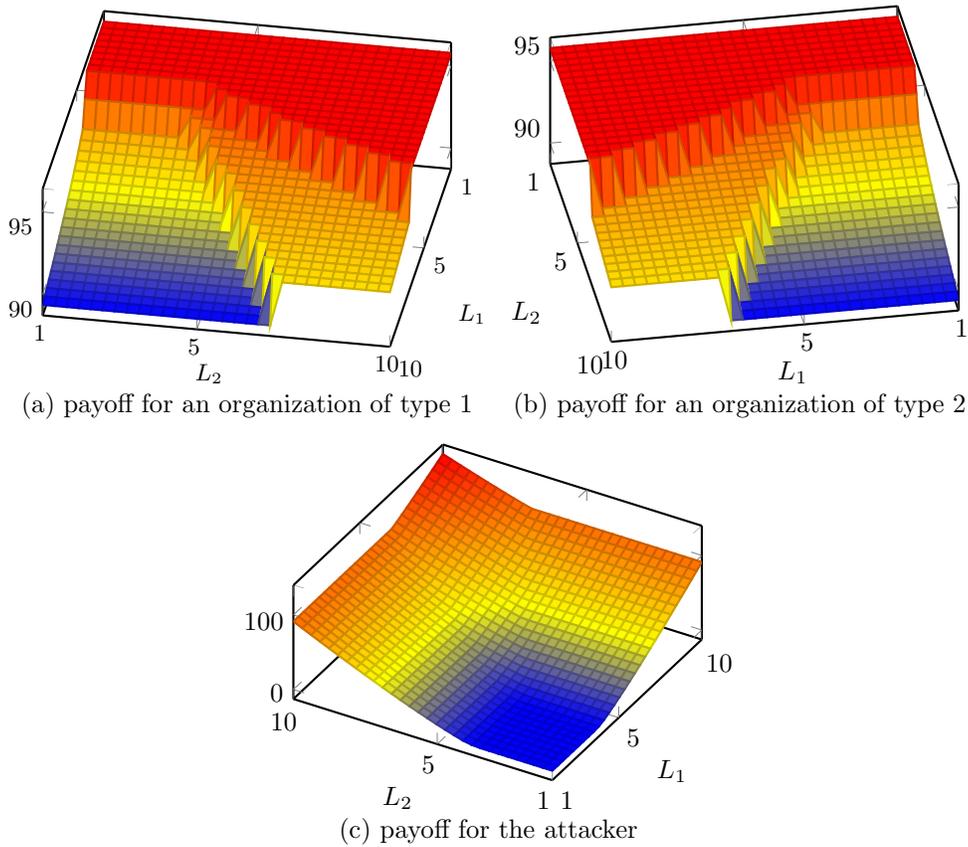

Figure~\ref{fig:L} shows the payoffs of individual organizations from the two groups as well as the attacker, for various values of the costs $L_1$ and $L_2$ of permanent data loss.
As expected, we see from the attacker's payoff (Figure~\ref{fig:L}(c)) that as loss costs increase, organizations become more willing to pay higher ransoms, so the attacker's payoff increases.
On the other hand, we observe a more interesting phenomenon in the organizations' payoffs.
As the loss cost (e.g., $L_1$) of one group (e.g., group 1) increases, the payoff of organizations in that group (e.g., Figure~\ref{fig:L}(a)) decreases.
However, we also see an \emph{increase} in the payoff (e.g., Figure~\ref{fig:L}(b)) of organizations in the \emph{other} group (e.g., group 2).
The reason for this increase is in the strategic nature of attacks: as organizations in one group become more attractive targets, attackers are more inclined to focus their efforts on this group, which results in lower intensity attacks against the other group.
This substitution effect, which can be viewed as a negative externality between groups of organizations, is strong when the attacker's efforts are focused (e.g., when ransomware is deployed using spear-phishing campaigns).

%% file: related.tex
\section{Related Work}
\label{sec:related}



\textbf{Ransomware:} Early work by Luo and Liao, in 2007 and 2009, respectively, represented first exploratory analyses of the ransomware phenomenon \cite{luo2007awareness,luo2009ransomware}. They focus on increased awareness (in particular, by employees) as a major means to diminish the effectiveness of ransomware attacks, which is a key recommendation regarding ransomware mirrored in the 2017 Verizon DBIR report ten years later: ``stress the importance of software updates to anyone who'll listen \cite{Verizon17}.''

In 2010, Gazet investigated the quality of code, functionalities and cryptographic primitives of 15 samples of ransomware~\cite{gazet2010comparative}. The studied sample of ransomware malware was quite basic from the perspective of programming quality and sophistication, and did not demonstrate a high level of thoroughness regarding the application of cryptographic primitives. However, the analysis also showed the ability to mass propagate as a key feature.

Highlighting ransomware's increasing relevance, Proofpoint reported that 70\% of all malware encountered in the emails of its customer base during a 10-month interval in 2016 was ransomware. At the same time, the same company reported that the number of malicious email attachments grew by about 600\% in comparison to 2015 \cite{Proofpoint16}. In addition, many modern forms of ransomware have worm capabilities as demonstrated in a disconcerting fashion by the 2017 WannaCrypt/WannaCry attack, which affected 100,000s of systems of individual users and organizations leading even to the breakdown of industrial facilities.

Other studies also focus on providing practical examples of and empirical data on ransomware including work by O'Gorman and McDonald~\cite{ogorman2012ransomware}, who provide an in-depth perspective of specific ransomware campaigns and their financial impact. In a very comprehensive fashion, Kharraz et al. analyze ransomware code observed in the field between 2006 and 2014~\cite{kharraz2015cutting}. Their results mirror Gazet's observation on a much broader pool of 1,359 ransomware samples, i.e., currently encountered ransomware lacks complexity and sophistication. 

Nevertheless, it causes major harm. To cite just a few figures, the total cost of ransomware attacks (including paid ransoms) increased to \$209 Million for the first three months in 2016 according to FBI data. In comparison, for 2015 the FBI only reported damages of about \$24 Million \cite{Reuters16}. 

Drawing on their earlier research, Kharraz et al. developed practices to stem the impact of ransomware. Their key insight is that ransomware needs to temper with user or business data, so that increased monitoring of data repositories can stop ransomware attacks while they unfold, and detect novel attacks that bypassed even sophisticated preventative measures \cite{kharraz2016unveil}. Scaife et al. also present an early-warning detection approach regarding suspicious file activity~\cite{scaife2016cryptolock}. 

Extending this line of work to a new context, Andronio et al. study ransomware in the mobile context and develop an automated approach for detection while paying attention to multiple typical behaviors of ransomware including the display of ransom notes on the device's screen ~\cite{andronio2015heldroid}. Likewise, Yang et al. focus on mobile malware detection, and specifically ransomware identification~\cite{yang2015automated}.

Ransomware attacks appear to be predominantly motivated by financial motives, which supports the usage of an economic framework for their analysis. However, other related types of attacks such as malicious device bricking (see, for example, the BrickerBot attack focusing on IoT devices \cite{Mcafee17}) may be based on a purely destructive agenda, with less clearly identifiable motives.

While knowledge about the technical details and financial impact of ransomware is growing, we are unaware of any research which focuses on the strategic economic aspects of the interactions between cybercriminals that distribute ransomware and businesses or consumers who are affected by these actions.

\textbf{Economics of Security:} Game-theoretic models to better understand security scenarios have gained increased relevance due to the heightened professionalism of cybercriminals. 
Of central interest are models that capture interdependencies or externalities arising from actions by defenders or attackers \cite{Laszka_survey}. 

A limited number of research studies focus on the modeling of the attack side. For example, Schechter and Smith capture different attacker behaviors \cite{Schechter03}. In particular, they consider the cases of serial attacks where attackers aim to compromise one victim after another, and the case of a parallel attack, where attackers can automate their attacks to focus on multiple defenders at one point in time. We follow the latter approach, which has high relevance for self-propagating ransomware such as WannaCrypt/WannaCry.

Another relevant aspect of our work are incentives to invest in backup technologies, which have found only very limited consideration in the literature. Grossklags et al. investigate how a group of defenders individually decide on how to split security investments between preventative technologies and recovery technologies (called self-insurance) \cite{grossklags2008secure}. In their model, preventative investments are subject to interdependencies drawing on canonical models from the literature on public goods \cite{Varian04}, while recovery investments are effective independent from others' choices. Fultz and Grossklags \cite{Fultz09} introduce strategically acting attackers in this framework, who respond to preventative investments by all defenders. In our model, backup investments are also (partially) effective irrespective of others' investment choices. However, in the context of ransomware, pervasive investments in backup technologies can have a deterrence and/or displacement effect on attackers \cite{Becker68}, which we capture with our work.

While we draw on these established research directions, to the best of our knowledge, our work is the first game-theoretic approach focused on ransomware.

%% file: concl.tex
\section{Concluding Remarks}
\label{sec:concl}

In this paper, we have developed a game-theoretic model, which is focused on key aspects of the adversarial interaction between organizations and ransomware attackers. In particular, we place significant emphasis on the modeling of security investment decisions for mitigation, i.e., level of backup effort, as well as the strategic decision to pay a ransom or not. 

These factors are interrelated and also influence attacker behavior. For example, in the context of kidnappings by terrorists it has been verified based on incident data that negotiating with kidnappers and making concessions encourages substantially more kidnappings in the future \cite{Brandt16}. We would expect a similar effect in the context of ransomware, where independently acting organizations who are standing with the ``back against the wall'' have to make decisions about ransom payments to get operations going again, or to swallow the bitter pill of rebuilding from scratch and not giving in to cybercriminals. Indeed, our analysis shows that there is a sizable gap between the decentralized decision-making at equilibrium and the socially optimal outcome. This raises the question whether organizations paying ransoms should be penalized? However, this (in turn) poses a moral dilemma, for example, when patient welfare at hospitals or critical infrastructure such at power plants are affected \textit{now}.

An alternative pathway is to (finally) pay significantly more attention to backup efforts as a key dimension of overall security investments. The relative absence of economic research focused on optimal mitigation and recovery strategies is one key example of this omission. A laudable step forward is the recently released factsheet document by the U.S. Department of Health \& Human Service on ransomware and the Health Insurance Portability and Accountability Act (HIPAA) \cite{Hipaa17}. It not only states that the encryption of health data by ransomware should be considered a security breach under HIPAA (even though no data is exfiltrated\footnote{The reasoning is as follows: ``When electronic protected health information (ePHI) is encrypted as the result of a ransomware attack, a breach has occurred because the ePHI encrypted by the ransomware was acquired (i.e., unauthorized individuals have taken possession or control of the information), and thus is a ``disclosure'' not permitted under the HIPAA Privacy Rule. \cite{Hipaa17}''}), but also that having a data backup plan is a required security effort for all HIPAA covered organizations.

An interesting question for future research is the role of cyberinsurance in the context of ransomware, i.e., specifically policies including cyber-extortion. How would these policies have to be designed to achieve desirable outcomes? As discussed above, in the case of kidnappings one would worry about incentivizing future kidnappings by making concessions via kidnapping insurance \cite{Fink14}; however, the design space in the context of ransomware is significantly more complex, but also offers more constructive directions.


%% file: proof.tex
\section{Proofs}
\label{sec:proof}

\subsection{Proof of Lemma~\ref{lem:paymentBestResponse}}
\label{app:proof1}
From Equation~\eqref{eq:payoffCompromisedOrg}, we have that the best-response strategy $p_i^*$ of organization $i$ is
\begin{align}
p_i^* \in& \argmax_{p \in \{0, 1\}} \left[ W_j - C_B \cdot b_i - \beta \left( \frac{F_j + (1 - p) \cdot L_j}{b_i} + T_j + p \cdot r \right) \right] \\
&= \argmax_{p \in \{0, 1\}} p \cdot \left( \frac{L_j}{b_i} - r \right) .
\end{align}
Clearly, $p_i^* = 1$ is a best response if and only if $\frac{L_j}{b_i} - r \geq 0$,
and $p_i^* = 0$ is a best response if and only if $\frac{L_j}{b_i} - r \leq 0$. $\qed$

\subsection{Proof of Lemma~\ref{lem:backupBestResponseNoAttack}}
\label{app:proof2}

From Equation~\eqref{eq:payoffSecureOrg}, we have that the best-response strategy $b_i^*$ of organization $i$ is
\begin{equation}
b_i^* \in \argmax_{b_i \in \Real_+} \left[  W_j - C_B \cdot b_i - \beta \frac{F_j}{b_i} \right] 
	.
\end{equation}
To find the maximizing $b_i^*$, we take the first derivative of the payoff, and set it equal to $0$:
\begin{align}
-C_B + \beta \frac{F_j}{{b_i^*}^2} = 0 \\
b_i^* = \pm \sqrt{\beta\frac{F_j}{C_B}} ,
\end{align}
Since $b_i \in \Real_+$, the only local optima is $b_i^* = \sqrt{\beta\frac{F_j}{C_B}}$.
Further, the payoff is a concave function of $b_i$ as the second derivative is negative, which means that this $b_i^*$ is the global optimum and, hence, a unique best response. $\qed$

\ifExtendedVersion		
\subsection{Proof of Lemma~\ref{lem:backupBestResponseAttack}}
\label{app:proof3}

	We can express the expected payoff of organization $i$ as
	\begin{align*}
	\E&\left[\calU_O\right] = \! \left(\! 1 - V_j(a_1, a_2) \!\right) \!
	\left[ W_j - C_B \cdot b_i - \beta \frac{F_j}{b_i} \right] \nonumber\\
	& + V_j(a_1, a_2)
	\left[ W_j - C_B \cdot b_i - \beta \left( \frac{F_j + (1 - p_i^*(b_i)) \cdot L_j}{b_i}+ T_j + p_i^*(b_i) \cdot r \right) \right] \\
	=& W_j - C_B \cdot b_i - \beta\frac{F_j}{b_i}  - \beta V_j(a_1, a_2) 
	\left[ \frac{(1 - p_i^*(b_i)) \cdot L_j}{b_i} + T_j + p_i^*(b_i) \cdot r \right] .
	\end{align*}
	Notice that we incorporated the second-stage best-response payment strategy $p_i^*(b_i)$, which is given by Lemma~\ref{lem:paymentBestResponse} and depends on the backup strategy $b_i$.
	
	We find the best response $b_i^*$ by dividing the search space into two regions, $b_i \leq \frac{L_j}{r}$ and $b_i \geq \frac{L_j}{r}$, and find the best strategies in both regions.
	
	\noindent\textbf{Case 1:}
	First, consider $b_i \leq \frac{L_j}{r}$, which implies that $p_i^*(b_i) = 1$.
	In this case, the best backup strategy $b_i^* \leq \frac{L_j}{r}$ is
	\begin{align}
	b_i^* \in& \argmax_{b_i \leq \frac{L_j}{r}} \left[W_j - C_B \cdot b_i - \beta\frac{F_j}{b_i} - \beta V_j(a_1, a_2)  
	\left(T_j + r \right)\right] \\
	&= \argmax_{b_i \leq \frac{L_j}{r}} \left[- C_B \cdot b_i - \beta\frac{F_j}{b_i} \right] .
	\end{align}
	Notice that this is equivalent to the maximization problem considered in Lemma~\ref{lem:backupBestResponseNoAttack}, whose solution was $\sqrt{\beta\frac{F_j}{C_B}}$.
	From the concavity of the objective, it follows that the best strategy $b_i^* \leq \frac{L_j}{r}$ is $\min\left\{\sqrt{\beta\frac{F_j}{C_B}}, \frac{L_j}{r}\right\}$.
	
	\noindent\textbf{Case 2:}
	Second, consider $b_i \geq \frac{L_j}{r}$.
	In this case, the best backup strategy $b_i^* \geq \frac{L_j}{r}$ can be expressed as
	\begin{align}
	b_i^* \in& \argmax_{b_i \geq \frac{L_j}{r}} \left[ W_j - C_B \cdot b_i - \beta\frac{F_j}{b_i}  - \beta V_j(a_1, a_2) \left( \frac{L_j}{b_i} + T_j \right) \right] \\
	&= \argmax_{b_i \geq \frac{L_j}{r}} \left[ - C_B \cdot b_i - \frac{\beta}{b_i} \left( F_j +  V_j(a_1, a_2)  L_j \right) \right] .
	\end{align}
	Note that the equation is correct even when $b_i = \frac{L_j}{r}$ (which implies $p_i^* = 1$),  since we can substitute $\frac{L_j}{b_i}$ for $r$.
	
	To find the best strategy $b_i^* \geq \frac{L_j}{r}$, we take the first derivative of the objective and set it equal to $0$:
	\begin{align}
	0 &= -C_B + \frac{\beta}{{b_i}^2} \left( F_j +  V_j(a_1, a_2)  L_j \right) \\
	b_i &= \pm \sqrt{\beta\frac{\left( F_j +V_j(a_1, a_2)  L_j \right)}{C_B}} .
	\end{align}
	Further, it is clear that the objective concave as the first derivative is strictly decreasing.
	Combining with the constraint $b_i \geq \frac{L_j}{r}$, this implies that the best strategy $b_i^* \geq \frac{L_j}{r}$ is $\max\left\{\sqrt{\beta\frac{\left( F_j +  V_j(a_1, a_2)  L_j \right)}{C_B}}, \frac{L_j}{r}\right\}$.
	
	Let's define $b_j^{\text{low}} = \sqrt{\beta\frac{F_j}{C_B}}$ and $b_j^{\text{high}} = \sqrt{\beta\frac{\left( F_j + V_j(a_1, a_2)  L_j \right)}{C_B}}$. Note that it is obvious that $b_j^{\text{low}} \leq b_j^{\text{high}}$. If we have $b_j^{\text{low}} > \frac{L_j}{r}$, the organization's payoff is increasing in Case 1 and for Case 2, it is increasing in $[\frac{L_j}{r}, b_j^{\text{high}}]$ and decreasing in $[b_j^{\text{high}}, \infty]$. Hence, the organization's payoff is maximized in $b_j^{\text{high}}$. In a similar way, the organization's payoff is maximized at $b_j^{\text{low}}$ when $b_j^{\text{high}} < \frac{L_j}{r}$. Otherwise, the organization's payoff is increasing in $[0,b_j^{\text{low}}]$, decreasing in $[b_j^{\text{low}}, \frac{L_j}{r}]$, increasing in $[\frac{L_j}{r}, b_j^{\text{high}}]$, and decreasing in $[b_j^{\text{high}}, \infty]$. Therefore, $b_i^* \in \left\{b_j^{\text{low}}, b_j^{\text{high}}\right\}$ (the one that gives the higher payoff or both if the resulting payoffs are equal). $\qed$
\fi

\subsection{Proof of Lemma~\ref{lem:ransomBestResponse}}
\label{app:proof4}

	The best-response ransom demand $r^*$ is
	\begin{align}
	r^* \in& \argmax_{r \in \Real_+} \left[\sum_j \sum_{i \in G_j} V_j(a_1, a_2) \cdot r \cdot p_i^*(r) \right] - C_A \cdot (a_1 + a_2) - C_D \\
	&= \argmax_{r \in \Real_+} \sum_j \sum_{i \in G_j} V_j(a_1, a_2) \cdot r \cdot 1_{\left\{r \leq \frac{L_j}{\hat{b}_j}\right\}}  \\
	&= \argmax_{r \in \Real_+} \sum_j \left|G_j\right| \cdot V_j(a_1, a_2) \cdot r \cdot 1_{\left\{r \leq \frac{L_j}{\hat{b}_j}\right\}} .
	\end{align}
	Clearly, the optimum is attained at either $\frac{L_1}{\hat{b}_1}$ or $\frac{L_2}{\hat{b}_2}$.
	Since we assumed that $\frac{L_1}{\hat{b}_1} \leq \frac{L_2}{\hat{b}_2}$,
	we have that $r = \frac{L_1}{\hat{b}_1}$ is a best response if and only if
	\begin{align}
	\left(\left|G_1\right|  V_1(a_1, a_2) + \left|G_2\right|  V_2(a_1, a_2)\right) \frac{L_1}{\hat{b}_1} \geq \left|G_2\right|  V_2(a_1, a_2) \frac{L_2}{\hat{b}_2} \\
	\left|G_1\right|  V_1(a_1, a_2) \frac{L_1}{\hat{b}_1} \geq \left|G_2\right|  V_2(a_1, a_2) \left( \frac{L_2}{\hat{b}_2} -  \frac{L_1}{\hat{b}_1} \right) .
	\end{align}
	Further, an analogous condition holds for $r = \frac{L_2}{\hat{b}_2}$ being a best response, which concludes our proof. $\qed$

\subsection{Proof of Lemma~\ref{lem:AttBREffort}}
\label{app:proof5}

Recall that the attacker's expected payoff is 
\begin{equation}
\E\left[\calU_A\right] = \left(\sum_j \sum_{i \in G_j} V_j(a_1, a_2)\cdot p_i \cdot r \right) - C_A \cdot (a_1 + a_2) - C_D \cdot 1_{\left\{a_1 > 0 \text{ or } a_2 > 0\right\}} .
\end{equation}
Consider that $a_1 + a_2 = a_{\text{sum}}$ and $r$ are given, and $a_{\text{sum}} > 0$.
Under these conditions, the attacker's best strategy is
\begin{align}
a_1^* \in& \argmax_{a_1 \geq 0} \left(\sum_j \sum_{i \in G_j} V_j(a_1, a_2)\cdot p_i^*(r) \cdot r \right) - C_A \cdot (a_1 + a_2) - C_D \\
&= \argmax_{a_1 \geq 0} \frac{a_1}{D + a_{\text{sum}}} |G_1| \cdot 1_{\left\{r \leq \frac{L_1}{\hat{b}_1}\right\}} + \frac{a_{\text{sum}} - a_1}{D + a_{\text{sum}}} |G_2| \cdot 1_{\left\{r \leq \frac{L_2}{\hat{b}_2}\right\}},
\end{align}
giving the non-negative payoff. The best strategy can be calculated readily.  $\qed$

\subsection{Proof of Proposition~\ref{pro:NE}}
\label{app:proofPro1}

Lemma~\ref{lem:AttBREffort} shows the attacker's best-response attack effort for fixed effort level, i.e., $a_{sum}$. In this Lemma, for example, $a_1^*=0$ and $a_2^*=a_{sum}$ is the attacker's best-response effort if $|G_1| \cdot 1_{\left\{r \leq \frac{L_1}{\hat{b}_1}\right\}} < |G_2| \cdot 1_{\left\{r \leq \frac{L_2}{\hat{b}_2}\right\}}$ and the resulting attacker's payoff is non-negative. 
According to Lemma~\ref{Lem:AttBR-Ransom}, the attacker's best-response ransom demand is either $\frac{L_1}{\hat{b}_1}$ or $\frac{L_2}{\hat{b}_2}$ and without loss of generality, we have assumed that  $\frac{L_1}{\hat{b}_1} \leq\frac{L_2}{\hat{b}_2}$.

For this case, the attacker's payoff is equal to:

\begin{equation}
\E\left[\calU_A\right] = \frac{a_{sum}}{D+a_{sum}} |G_2| \cdot r \cdot 1_{\left\{r \leq \frac{L_2}{\hat{b}_2}\right\}} - C_A \cdot a_{sum} - C_D.
\label{eq:AttPayoff}
\end{equation}

If the above equation is negative, i.e., $$r < \frac{ \left(D + a_{sum}\right) \left(C_A \cdot a_{sum} + C_D\right)}{a_{sum} \cdot |G_2| \cdot 1_{\left\{r \leq \frac{L_2}{\hat{b}_2}\right\}}},$$ the attacker's best-response effort is $a_1^* = a_2^*=0$. To satisfy the above condition, we replace $r$ with $\frac{L_2}{\hat{b}_2}$, which gives $$ \frac{L_2 \cdot a_{sum} \cdot |G_2| \cdot 1_{\left\{r \leq \frac{L_2}{\hat{b}_2}\right\}}}{L_2\left(D + a_{sum}\right) \left(C_A \cdot a_{sum} + C_D\right)} < \hat{b}^*_2.$$

Further, the defender's best-response backup strategy when there is no attack, i.e., $a_1^*=a_2^*=0$ is calculated based on Lemma~\ref{lem:BackupNoAtt}. By inserting the value of $\hat{b}^*_2$ from Lemma~\ref{lem:BackupNoAtt}, we can readily have the following: $$ \frac{L_2 \cdot a_{sum} \cdot |G_2| \cdot 1_{\left\{r \leq \frac{L_2}{\hat{b}_2}\right\}}}{L_2\left(D + a_{sum}\right) \left(C_A \cdot a_{sum} + C_D\right)} < \sqrt{\beta\frac{F_2}{C_B}}. $$

Another condition can be calculated similarly. $\qed$

\ifExtendedVersion
\section{Further Numerical Illustrations}
\label{app:Num}

\pgfplotsset{figF/.style={%
width=0.55\linewidth,
xlabel={$F_1$},
ylabel={$F_2$},
thick,
xtick={1, 5, 10},
ytick={1, 5, 10},
xlabel shift=-0.5em,
ylabel shift=-0.5em,
zlabel shift=-0.5em,
}}

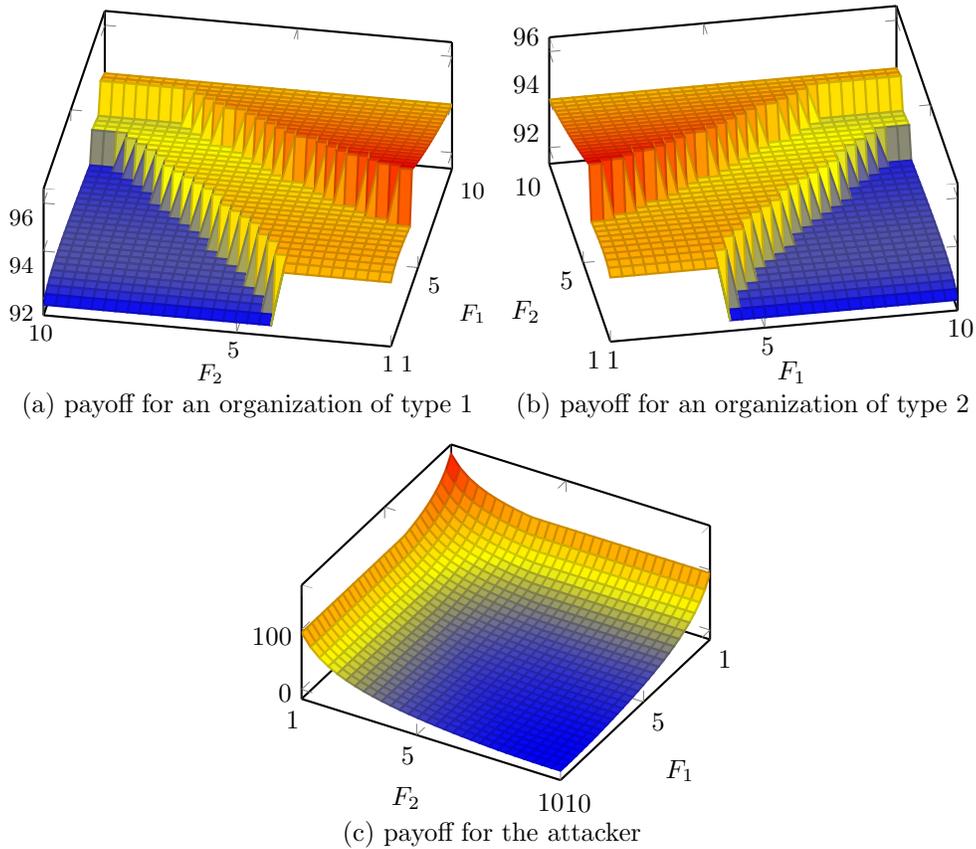
\begin{figure}[h!]
\centering
\renewcommand*{\arraystretch}{0}
\setlength{\tabcolsep}{1.5pt}
\begin{tabular}{cc}
\begin{tikzpicture}[font=\small]
\begin{axis}[
figF,
view={280}{55},
]
\addplot3 [surf, mesh/ordering=x varies, mesh/rows=30, mesh/cols=30] table[x=F1, y=F2, z=NE_org1_payoff, col sep=comma, comment chars={\%}] {results/F.csv}; 
\end{axis}
\end{tikzpicture}
&
\begin{tikzpicture}
\begin{axis}[
figF,
view={350}{55},
]
\addplot3 [surf, mesh/ordering=x varies, mesh/rows=30, mesh/cols=30] table[x=F1, y=F2, z=NE_org2_payoff, col sep=comma, comment chars={\%}] {results/F.csv}; 
\end{axis}
\end{tikzpicture}
\\
(a) payoff for an organization of type 1 & (b) payoff for an organization of type 2 \\[1.2em]
\multicolumn{2}{c}{
\begin{tikzpicture}
\begin{axis}[
figF,
view={120}{55},
]
\addplot3 [surf, mesh/ordering=x varies, mesh/rows=30, mesh/cols=30] table[x=F1, y=F2, z=NE_att_payoff, col sep=comma, comment chars={\%}] {results/F.csv}; 
\end{axis}
\end{tikzpicture}
}
\\
\multicolumn{2}{c}{(c) payoff for the attacker}
\end{tabular}
\caption{The expected payoff of individual organizations of (a) type 1 and (b) type 2 as well as (c) the attacker in equilibrium for various costs of random-failure data losses $F_1$ and~$F_2$.}
\label{fig:F}
\end{figure}

Figure~\ref{fig:F} shows the payoffs of individual organizations from two groups as well as the attacker, for various values of the costs $F_1$ and $F_2$ of random-failure data losses.
Figure~\ref{fig:F}(c) shows that the attacker's payoff \emph{decreases} as these costs increase, which is in contrast with the cost-payoff relationship shown by Figure~\ref{fig:L}(c).
The reason for this decrease is that organizations spend more on backups when the cost of random-failure data losses is higher, which in turn reduces the potential loss from ransomware attacks and, hence, the ransom that an attacker may demand.

We see more interesting phenomena in the organizations' payoffs.
For most combinations of cost values, as expected, the payoff of an organization decreases gradually as the random-failure data loss cost in its group increases.
However, for certain combinations of cost values (e.g., $F_1 = 4$ and $F_2 = 1$), increasing the cost of one group (e.g., group 1) \emph{increases} the payoff of organizations in that group (see, e.g., Figure~\ref{fig:F}(a)) but \emph{reduces} the payoff in the other group.
This surprising phenomenon is due to the incentives that increased costs of random-failure data losses provide for organizations to invest more in backups.
At certain tipping points, such as $F_1 = 4$ and $F_2 = 1$, even a small change in the costs can lead to a significantly different equilibrium, in which the focus of the attacker is shifted from one group to another.
This substitution effect is very similar to the one exhibited in Figure~\ref{fig:L}.

\fi

